\documentclass[10pt, conference, compsocconf]{IEEEtran}

\usepackage{changepage}
\usepackage{amssymb}
\usepackage{amsmath}
\usepackage{mathrsfs}
\usepackage{amsthm}
\usepackage{bm}
\usepackage{multirow}
\usepackage{booktabs}
\usepackage{graphicx}
\usepackage{array}
\usepackage{stfloats}
\usepackage[dvipsnames, svgnames, x11names]{xcolor}
\usepackage{color}
\usepackage{tikz}
\usepackage[symbol]{footmisc}
\usepackage[numbers,sort&compress]{natbib}
\usepackage{bm}
\usepackage{enumitem}
\usepackage{footmisc}

\usepackage{enumitem}
\setlistdepth{20}
\renewlist{itemize}{itemize}{20}
\setlist[itemize]{label=$\cdot$}
\setlist[itemize,1]{label=\textbullet}
\setlist[itemize,2]{label=--}
\setlist[itemize,3]{label=*}

\usepackage{url}

\usepackage{breakurl}
\usepackage{xurl}

\newcommand{\Contract}{\mathscr{C}}
\newcommand{\PoK}{\mathsf{PoK}}

\newcommand{\VPKE}{\mathsf{VPKE}}

\newcommand{\huaA}{\mathcal{A}}
\newcommand{\huaD}{\mathcal{D}}

\newcommand{\huaF}{\mathcal{F}}
\newcommand{\huaL}{\mathcal{L}}
\newcommand{\huaR}{\mathcal{R}}
\newcommand{\huaH}{\mathcal{H}}
\newcommand{\huaP}{\mathcal{P}}
\newcommand{\huaW}{\mathcal{W}}
\newcommand{\huaT}{\mathcal{T}}
\newcommand{\huaG}{\mathcal{G}}
\newcommand{\huaS}{\mathcal{S}}
\newcommand{\reward}{\mathsf{Quality}}

\newcommand{\KGen}{\mathsf{KeyGen}}
\newcommand{\Enc}{\mathsf{Enc}}
\newcommand{\Dec}{\mathsf{Dec}}
\newcommand{\hash}{\mathsf{h}}
\newcommand{\comm}{\mathsf{comm}}

\newcommand{\commit}{\mathsf{Commit}}
\newcommand{\open}{\mathsf{Open}}

\newtheorem{theorem}{Theorem}
\newtheorem{lemma}{Lemma}

\newcommand{\system}{$\mathsf{Dragoon}$}
\newcommand{\zebralancer}{$\mathsf{ZebraLancer}$}

\newcommand{\aux}{\mathsf{param}}

\newcommand{\PoQoEA}{\mathsf{PoQoEA}}

\newcommand{\Verify}{\mathsf{VerifyPKE}}

\newcommand{\Prove}{\mathsf{ProvePKE}}
\newcommand{\cProve}{\mathsf{ProveQuality}}
\newcommand{\cVerify}{\mathsf{VerifyQuality}}

\newcommand{\ignore}[1]{}
\newcommand{\yuan}[1]{\textcolor{red}{#1}}

\def\bitcoin{%
	\leavevmode
	\vtop{\offinterlineskip 
		\setbox0=\hbox{B}%
		\setbox2=\hbox to\wd0{\hfil\hskip-.03em
			\vrule height .3ex width .15ex\hskip .08em
			\vrule height .3ex width .15ex\hfil}
		\vbox{\copy2\box0}\box2}}

\begin{document}

\title{\system: Private Decentralized HITs Made Practical}

\author{\IEEEauthorblockN{Yuan Lu$^1$,  Qiang Tang$^{1,2}$, Guiling Wang$^1$}
\IEEEauthorblockA{$^1$New Jersey Institute of Technology, $^2$JDD-NJIT-ISCAS Joint Blockchain Lab\\
Email: \{yl768, qiang, gwang\}@njit.edu}
}


\maketitle

\renewcommand*{\thefootnote}{\arabic{footnote}}

\begin{abstract}
With the rapid popularity of blockchain, 
decentralized human intelligence tasks (HITs) are proposed  to crowdsource human knowledge without  relying on   vulnerable third-party platforms.
However,  the inherent limits of blockchain cause decentralized HITs to face a few ``new'' challenges. For example,
the confidentiality of solicited data turns out to be the {\em sine qua non}, though it was an   arguably dispensable  property in the centralized setting. 
To ensure the ``new'' requirement of data privacy,
existing decentralized HITs use generic zero-knowledge proof frameworks (e.g.,  SNARK), 
but   scarcely perform well in practice, due to the inherently expensive cost of generality.

We present a {\em practical} decentralized protocol for HITs, which  also achieves the {\em fairness} between requesters and workers. 
At the core of our contributions, we avoid the powerful yet highly-costly {\em generic} zk-proof tools    
and propose a special-purpose scheme to prove the quality of encrypted data.
%
%
By various   non-trivial statement reformations, 
proving the quality of encrypted data is   reduced to efficient verifiable decryption, 
thus making decentralized HITs practical.
Along the way, we rigorously define the ideal functionality of decentralized HITs and then  prove the security  due to the ideal/real paradigm.

We further instantiate our protocol to implement a system called \system{}\footnote{In German history, a dragoon was a lancer that  was particularly light and firearmed.  As an analog, \system{}   is a super lightweight  and highly robust system that enables the modern freelancers to enjoy     decentralized HITs.}, an instance of which 
is deployed atop Ethereum  to facilitate an image annotation task used by ImageNet. Our evaluations demonstrate its   practicality: the  on-chain handling cost of \system{} is even less than the handling fee of Amazon's Mechanical Turk for the same ImageNet HIT.

\end{abstract}



\section{Introduction}\label{introduction}

Crowdsourcing empowers open collaborations over the Internet.
A remarkable case is to gather knowledge through \underline{h}uman \underline{i}ntelligence \underline{t}asks (HITs).
In HITs, a {\em requester} specifies a few questions which some {\em workers} can answer, 
such that the requester obtains answers and the workers get paid.
Since HITs were firstly minted in Amazon's MTurk \cite{MTurk}, they have been widely adopted, e.g., 
to build training datasets for machine learning  \cite{feifei,ng,video}.
In particular, ImageNet \cite{Imagenet}, an impactful deep learning  benchmark, 
was created through thousands of HITs  and laid stepping stones for the deep learning paradigm.

Nevertheless, 
both   academia and industry  \cite{ABI13, sdhc, mturkbotpanic,PWC15,mturkgolden,SZ15,imagenet-details,mturkbot,turkopticon,MCN16}
realize   the  broader adoption of HITs  is severely impeded in practice, as a result of the serious  security concerns
of {\em free-riding} and {\em false-reporting}:
(i) on the one hand, HITs suffer from low-quality answers, as misconducting workers or even bots would try to reap rewards without making real efforts \cite{sdhc, mturkbotpanic}; 
(ii) on the other hand, quite many requesters in the wild arbitrarily reject answers in order to collect data without paying  \cite{turkopticon}
through manipulating some real-world guidelines set forth for the requester to reject low-quality answers \cite{PWC15,mturkgolden,SZ15}. 

%



{Free-riding} and {false-reporting}  become the  major obstacles to the broader adoption of HITs 
  participated by mutually distrustful users \cite{turkopticon},
and therefore raise a basic requirement of {fairness} in HITs, 
namely, the requester   pays a worker,    iff  the worker  puts forth  a {qualified}  answer. Many studies \cite{sdhc,mturkbotpanic,PWC15,SZ15,mturkgolden,ABI13,imagenet-details}  characterize the purpose and then design proper incentives and payment policies for the needed fairness.
%
%
%

Notwithstanding, most traditional   solutions to fairness  \cite{sdhc,mturkbotpanic,PWC15,SZ15,mturkgolden,ABI13,imagenet-details}   fully trust  in a de facto {\em centralized} third-party platform  to enforce the payment policies for the basic {fairness} requirement in HITs.  
Unfortunately, putting trust in a single party turns out to be    {\em vulnerable} and {\em elusive} in practice, as a reflection of tremendous compromises, outages and misfeasance of real-world crowdsourcing platforms  \cite{MCN16,turkopticon, WazeDown}.
For instance, one of the most popular crowdsourcing platform, MTurk,  allows corrupted requesters  to reap data without paying  \cite{MCN16,turkopticon}.  
Worse still, all well-known weaknesses of  overtrusted third-parties, such as single-point failure \cite{WazeDown} and tremendous privacy leakage \cite{Apple}   remain as serious vulnerabilities in the special case of crowdsourcing. 
Let alone the   third-party  platforms   impose   expensive handling charges. For example, MTurk charges a handling fee up to  45\% of   overall incentives \cite{mturkfee}.


\smallskip
\noindent
{\bf New challenges in decentralization.}
Recognizing those drawbacks of {centralized} crowdsourcing, recent attempts  \cite{zebralancer,duan2019aggregating} 
initiated the {\em decentralized} crowdsourcing through the newly emerged 
blockchain\footnote{Remark that through the paper, blockchain refers to permissionless blockchain  (e.g. Ethereum mainnet) 
that is open to any Internet node.} technology.
Their aim is to ``simulate'' a virtual platform that is trustable to enforce the payment policies, thus removing the vulnerable   centralized platforms.
%

However, as shown in the seminal studies   on the blockchain \cite{KZZ16,KMS16},    decentralization through the blockchain  also
brings forth a few ``new'' security challenges, which can render the incentive mechanisms of HITs completely ineffective \cite{zebralancer}.
%

\underline{\smash{\em Privacy as a basic requirement}}. 
	In particular, due to the transparency of blockchain \cite{KZZ16,KMS16}, once some answers are submitted,
	any malicious worker can simply copy and re-submit them to earn rewards without making any real efforts,
	which immediately allows free-riding and cracks the basic fairness of HITs. 
	Namely, the transparent blockchain presents all workers a new option: 
	running a simple automated script  to ``copy-and-paste'' other answers submitted to the blockchain, 
	which was infeasible in previous centralized systems. 
	Having the new option of free-riding, rational workers might wait to copy  instead of conducting any  efforts. 
	Sorta ``tragedy of the commons'' could occur, and eventually no one would respond with independent 
	answers \cite{hardin1968tragedy,stewart2017crowdsourcing,huberman2009crowdsourcing,david2001tragedy}. 
	That said, the straightforwardly decentralized crowdsourcing could lose all   basic utilities  and 
	fail  to gather anything meaningful!
	
	Therefore, to make the decentralized crowdsourcing systems   function as desired,   privacy   becomes an  {\em indispensable} requirement instead of an advanced bonus property.

\smallskip
\noindent
{\bf State-of-the-art \& open problem}. 
To overcome blockchain's inherent limits,
prior art \cite{zebralancer} proposes the general outsource-then-prove framework for {\em private} decentralized HITs. 
%
%
It enables the requester to prove the  quality of  answers that are encrypted to her, without revealing the actual answers.
Such a proof becomes the crux to ensure privacy and simultaneously deters false-reporting and free-riding.
In addition, the feasibility challenge sprouts up as the blockchain needs to verify the proof,
which means the proof size and verification cost must be small enough to meet the    limited on-chain   resources.

For above reasons, prior  work relies on some  {\em generic} zero-knowledge proof (zk-proof) framework that is succinct in proof size and efficient for verifying, in particular SNARK\footnote{Remark that though the rise of Intel SGX becomes a seemingly enticing alternative of SNARK to go beyond many limits of blockchain by remote attestations \cite{Ekiden}, unfortunately, recent Foreshadow attacks \cite{foreshadow} allow the adversary to forge ``attestations'' by stealing the attestation key hardcoded in any SGX Enclave, which
	seriously challenges the already heavy  assumption of ``trusted'' hardware, and makes it even more illusive to trust SGX in practice.}  \cite{qap,BCG13,pinocchio}
%
 to reduce the on-chain verification cost.

Nonetheless, generic zk-proofs such as SNARK inevitably inherit    low performance  for the convenience of achieving generality, causing that prior private decentralized HITs suffer  from  an unbearable off-chain proving cost  and a still significant on-chain verifying expense: 
\begin{itemize}
	
	\item \underline{\smash{{\em Infeasible proving} (off-chain)}}. The proving    of   {\em generic} zk-proofs (e.g.,  SNARK) seems inherently  complex, due to the burdensome  NP-reduction for  generality.
	In   particular,  prior study \cite{zebralancer-full} reported 56 GB memory and 2 hours are needed to prove 
	whether an encrypted answer coincides with the majority of all encrypted submissions at a very small scale, e.g.,  at most eleven answers. Such a performance prevents the previous protocol from being usable by any normal requesters using   regular PCs.

	\item \underline{\smash{{\em Costly verification}  (on-chain)}}. 
	Existing blockchains (e.g. Ethereum)  are   feasible to verify
	only few types of {\em generic} zk-proofs such as  SNARK,
	whose verification need to compute a dozen of  expensive pairings over elliptic curve \cite{qap,BCG13,pinocchio}.
	So the on-chain verification of these zk-proofs  is not only computationally costly, but also financially expensive. Currently in Ethereum, 12 pairings already spend $\sim$500k gas   \cite{EIP1108}, and verifying a SNARK proof costs even more (about half US dollar).
\end{itemize}

Given the insufficiencies of the  state-of-the-art,  the following critical problem remains open:
	\begin{center}
	\noindent{\em  How to design a practical private decentralized HITs protocol for crowdsourcing human knowledge?}
		\end{center}

\noindent
{\bf Our contributions}. 
To answer the above unresolved problem,
we 
present
a {\em practical} private decentralized HITs  protocol
for the major tasks of crowdsourcing human knowledge.
%
In sum, our core technical contributions are three-fold:

\begin{itemize}

	\item 
	To achieve  practical private decentralized HITs, 
	we  explore	various non-trivial optimizations to avoid  the cumbersome generic-purpose zk-proof framework,
	and reduce the protocol    to the specific verifiable decryption.
	As such, we attain concrete improvements by orders of magnitude, regarding both the proving and verification: 
	\begin{itemize}
		\item For proving, our approach is {\em two orders of magnitude} better than generic zk-proof.\footnote{Generic zk-proof   refers   zk-SNARK in our context, since  the only generic zk-proof that can be feasibly supported by existing blockchains is zk-SNARK.} 
		For the same HIT, the proving in our protocol costs  50 MB memory and 10 msec, while the generic proof  costs 10 GB    and 2 min.
		\item For   verifying, our result  improves upon the generic solution by   {\em nearly an order of magnitude}. The on-chain cost of verifying a proof for the quality of an answer to 106 batched binary questions is reduced to  $\sim$180k gas in Ethereum  (much smaller than verifying SNARK proofs) and typically   few US cents.
	\end{itemize}

	%
	
	
	%

	\item 
	We further  implement  our    protocol 	to instantiate a {\em practical} private decentralized crowdsourcing system  \system{},
	the handling cost of which could be even less than the existing centralized platforms such as MTurk. 
	
	\system{}
	is launched   atop Ethereum to conduct a typical HIT adopted by  ImageNet  \cite{imagenet-details} to solicit large-scale image annotations.
	To handle the  task, 
	\system{}   attains an  on-chain (handling) cost   $\sim$\$2 US dollars   at the time of writing. 
	In comparison, for the  same   task, the handling fee of MTurk   is at least \$4 currently  \cite{mturkfee,mturkprice}.
	
	Our result provides an insight  that the on-chain {\em handling fee} (characterizing the users' financial expense) in the  decentralized setting  
	can   approximate or even   less than the handling fee charged by centralized   platforms. 
	This   indicates     the {\em de facto} users can   financially benefit  from decentralization,
	though it is not contradictory to the common belief \cite{Sedgwick} that    decentralization  is more expensive w.r.t. the overall computational cost of the system.	
	

	\item Along the way,
	we firstly formulate the ideal functionality of decentralized HITs. The rigorous security model clearly defines what a HIT shall be and allows us to use the simulation-based paradigm to prove security against  subtle adversaries in the blockchain.
	
	In contrast,    existing decentralized HITs \cite{zebralancer,zebralancer-full} 
	have quite different property-based definitions on ``securities'', which at least makes the lack of  well-defined benchmark to compare them. Even worse, many of them are   ``flawed'', as failing  to capture all respects of the subtle  adversary in the  blockchain;
	for example, they allow  the corrupted requester to reap data without paying, if being given the standard ability of adversarially re-ordering message deliveries, 
	while our approach     precisely defines the security requirement against this  subtle attack.
	
	%

\end{itemize}

\noindent
{\bf Challenge \& our techniques}. 
The major challenge of making {\em private}  decentralized HITs practical is 
that the blockchain must learn the quality of some encrypted answers,
namely, to obtain some properties of what a few ciphertext are encrypting.
The state-of-the-art   \cite{zebralancer,zebralancer-full} proposed to reduce the problem  to generic zk-proofs,
by observing the requester can decrypt the answers, and then prove the quality of answers to the blockchain.
But this generic approach  incurs impractical expenses inherently, because  of the underlying heavyweight NP-reduction  for generality.

%
 
To conquer the above challenge,
we follow a different path that deviates from generic   zk-proofs to explore a concretely efficient solution.
At the core of our {\em private}  decentralized HITs protocol, we design a special-purpose non-interactive proof scheme to efficiently attest the quality of encrypted answers,
which removes heavyweight general-purpose zk-proofs and then avoids the inefficiency of generality.


\begin{figure}[!htp]
	\centering
	\includegraphics[width=9cm]{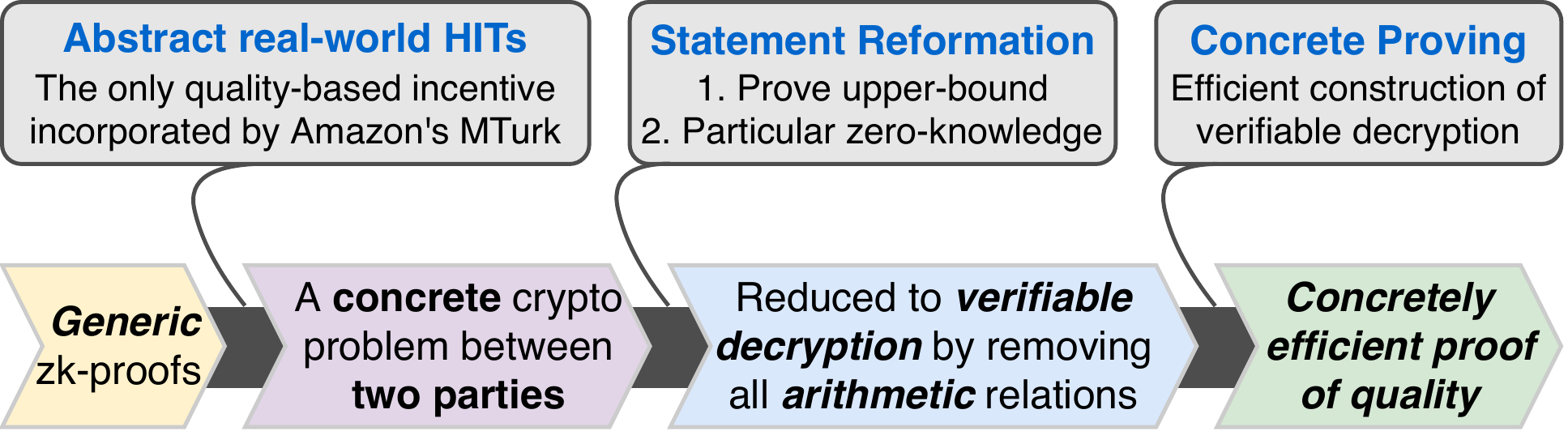}
	\caption{The path to  realizing efficient proofs for encrypted answers' quality.}
	\label{fig:optimize}
\end{figure}

The ideas behind our efficient proving scheme are a variety of  special-purpose optimizations to squeeze performance by removing needless generality, 
such that we reduce the problem of proving encrypted answers' quality from   generic-purpose zk-proof to particular verifiable decryption.
As shown in Fig  \ref{fig:optimize}, our core ideas are highlighted as:
\begin{itemize}
	\smallskip
	\item \underline{\smash{\em Abstracting real-world HITs}}. The first step is to well formulate an incentive mechanism  widely adopted  by  real-world HITs, namely, the {\em only} one incorporated by Amazon's MTurk \cite{mturkgolden}. Some golden standard  challenges (i.e., questions with known answers) \cite{sdhc} are mixed   with other  questions, so the  quality  of a worker is determined by her performance on the  golden standards.\footnote{This concrete incentive turns to be powerful, say it  can capture  most   HITs in Amazon's MTurk
	(c.f., the official   tutorial   \cite{mturkgolden}) and   also  adopted by the impactful   ImageNet   \cite{imagenet-details} to create large-scale deep learning benchmark.}

	We rigorously define the problem of proving the quality of encrypted answers for the above incentive mechanism.	
	So proving the quality of a worker is   reducible to a well-defined two-party problem:
	where the verifier needs to output the performance of the worker on a set of golden standard questions, given only a set of ciphertext answering these golden standards challenges.
	
	Nevertheless, solving this two-party problem is still challenging, 
	as it needs to compute the property of what a set of ciphertext are encrypting. 
	The generic version of the issue  falls into multi-input functional encryption \cite{goldwasser2014multi,boneh2015semantically}, which is well  known for its hardness and has no (nearly) practical solution so far.
	We thus   conduct the following   optimizations to further reduce the problem.
	
	%


	\item \underline{\smash{\em Statement reformation}}.
	The main obstacle  of removing the generic-purpose zk-proof framework is the arithmetic relations (i.e., some relationship unrepresentable in the algebraic domain).
	So	we   dedicatedly reform the statement of proving the quality of encrypted answers mainly in two ways
	 to remove all arithmetic relations.
	
	First, we let the   requester to prove the upper bound of each worker's quality instead of proving the exact number, which is a relaxation to the general cases, but does not abandon  any utility, since this property is enough to prevent any corrupted requester from paying less than what a worker deserves in our context where the reward is an increasing function of quality. 
	Second,  we realize that given the system's public knowledge,
	a tiny and constant portion of each worker's answer (i.e., the part answering gold standards) is already leaked,
	since this little portion is simulatable by the public knowledge; thus   we can explicitly  reveal   these ``already-leaked'' information.
		
	
	
	The above   reformations allow us to  reduce the problem of proving the quality of encrypted answers to standard  verifiable  encryption  without 
	 sacrificing securities/utilities.


	
	\item \underline{\smash{\em Concretely efficient proving scheme}}.
	Following the above optimizations,  the problem eventually is reduced to verifiable encryption,
	which becomes representable in concrete algebraic relations. 
	Along the way, we present a  certain  variant of verifiable decryption that is  concretely tailored for the  scenario of HITs where the plaintexts are short,
	and thus squeeze most performance out of it
	%
	and boost 
	private decentralized HITs practically.
	%
	
\end{itemize}
%


\section{Other related Work}\label{related}

Besides existing private decentralized HITs \cite{zebralancer-full,zebralancer}  discussed earlier, here we   briefly review  some pertinent  generic cryptographic frameworks  and  discuss their insufficiencies in the concrete  context of private decentralized crowdsourcing.

\smallskip
\noindent{\bf Privacy-preserving blockchain.}  A variety of studies \cite{KMS16,ZNP15,solidus} consider the general framework for privacy-preserving blockchain and smart contract. The approaches are powerful in the sense of their generality, yet are expensive for concrete use-cases in practice. 
For example, Hawk \cite{KMS16} leverages generic zk-proofs  to keep blockchain private, but incurs expensive proving expenses.
As such, it is unclear how to leverage these generic frameworks to design concretely efficient protocol for the special-purpose of crowdsourcing  \cite{zebralancer}.

\smallskip
 \noindent
{\bf Fair MPC using blockchain.} 
Decentralized crowdsourcing is a special-purpose fair MPC using blockchain. 
Kiayias, Zhou and Zikas \cite{KZZ16} consider the generic version of fair MPC in the presence of blockchain, but it is unclear how to adopt their generic  protocol in practice without expensively computational costs.
Recently, increasing interests focus on special-purpose variants of fair MPC in   aid of blockchain.
For example, \cite{bentov2017instantaneous,kumaresan2015use,david2018kaleidoscope} consider     poker games.
But  these  special-purpose solutions are over-tuned for  distinct   scenarios  and are unclear how to be used for private decentralized crowdsourcing.

\smallskip
 \noindent
{\bf Multi-input functional encryption.} The core problem of private decentralized crowdsourcing is to let the blockchain learn the quality of encrypted answers, which is straightforwardly reducible to    multi-input functional encryption (MIFE) \cite{goldwasser2014multi}. 
But   MIFE   relies on indistinguishability obfuscation   \cite{goldwasser2014multi} or
multi-linear maps \cite{boneh2015semantically}, which  currently we do not notice how to  instantiate under standard cryptographic
assumptions.


\section{Preliminaries}\label{preliminary}

Here we briefly review some relevant cryptographic notions. Following convention, we let $\overset{\$}{\leftarrow}$ to denote uniformly sampling and   $\approx_c$ to denote {\em computationally indistinguishable}.

\smallskip
\noindent
\textbf{Cryptocurrency ledger.} The cryptocurrency maintained atop the blockchain  instantiates a global bookkeeping ledger (e.g. denoted by $\huaL$) to  deal with ``coin'' transfers, transparently. It can be called out by an ideal functionality (i.e., a standard model of so-called smart contract \cite{KMS16,KZZ16}) as a subroutine to assist conditional payments. Formally,   cryptocurrency $\huaL$ can be seen as an ideal functionality interacting with a set of parties $\huaP=\{\huaP_i\}$   and the  adversary; it stores the balance $b_i$ for each $\huaP_i \in \huaP$, and handles the following oracle queries  \cite{KMS16, perun}:
\begin{itemize} 
	\item $\mathsf{Freeze Coins}$. On input $(\mathsf{freeze}, \mathcal{P}_i, b)$ from an ideal functionality $\huaF$ (i.e. a  smart contract), check whether $b_i \geq b$ and proceed as follows: if the check holds, let $b_i=b_i-b$ and $b_{\huaF} = b_{\huaF} + b$, send $(\mathsf{frozen}, \huaF, \mathcal{P}_i, b)$ to every entity; otherwise,  reply with $(\mathsf{nofund}, \mathcal{P}_i, b)$.
	\item $\mathsf{Pay Coins}$. On input $(\mathsf{pay}, \mathcal{P}_i, b)$ from an ideal functionality $\huaF$ (i.e.  a smart contract), check whether $b_{\huaF} \geq b$ and proceed as follows: if that is the case, let $b_i=b_i+b$ and $b_{\huaF} = b_{\huaF} - b$, send $(\mathsf{paid}, \huaF, \mathcal{P}_i, b)$ to every entity.
\end{itemize}

\smallskip
\noindent
\textbf{Commitment.}
The  commitment scheme is a two-phase protocol between a sender and a receiver. In the commit phase, a sender ``hides'' a string $\mathsf{msg}$ behind a   string $\comm$ with using a blinding $\mathsf{key}$, namely, the sender transmits $\comm = \commit(\mathsf{msg},\mathsf{key})$ to the receiver. 
In the reveal phase, the receiver gets  $\mathsf{key}'$ and   $\mathsf{msg}'$   as opening for $\comm$, and executes  $  \open(\comm,\mathsf{msg}',\mathsf{key}')$ to output    0 (reject) or (1) accept.
Through the paper, we consider \emph{computational hiding} and \emph{computational binding}. The former one requires the commitments of   two strings are computationally indistinguishable. The latter one means the receiver would not accept an opening to reveal  $\mathsf{msg}'\ne\mathsf{msg}$,
except with  negligible probability. 


\smallskip
\noindent
\textbf{Decisional Diffie-Hellman (DDH).}
 DDH problem is to tell that $d=ab$ or $d\overset{\$}{\leftarrow}\mathbb{Z}_p$, given $(g,g^a,g^b,g^d)$ where $a,b\overset{\$}{\leftarrow}\mathbb{Z}_p$ and $g$ is a generator of a cyclic group $\huaG$ of order $p$. 
 The DDH assumption states 
 $\{(g,g^a,g^b,g^d)\}\approx_c\{(g,g^a,g^b,g^{ab})\}$. 
 We assume  DDH assumption holds along with the paper.

\smallskip
\noindent
\textbf{Verifiable decryption.} 
We consider a specific verifiable public key encryption   ($\VPKE$) consisting of a tuple of algorithms  $(\KGen, \Enc, \Dec, \Prove, \Verify)$  with concrete verifiability to allow the decryptor to produce the plaintext along with a proof attesting the correct decryption \cite{camenisch2003practical}.

In short,   $\KGen$ can set up a pair of encryption-decryption algorithms $(\Enc_h, \Dec_k)$,  where $h$ and $k$ are public and private keys respectively. We let any $(\Enc_h, \Dec_k)\leftarrow\KGen(1^\lambda)$ to be a  public key encryption  scheme   satisfying semantic security. For presentation simplicity, we   also let $(\Enc_h, \Dec_k)$  denote the public-secret key pair $(h, k)$.
Moreover,  for any $(h, k)\leftarrow\KGen(1^\lambda)$, the $\Prove_k$ algorithm explicitly inputs the private key $k$ and the ciphertext $c$, and outputs a message $m$ with a proof $\pi$; the $\Verify_h$ algorithm explicitly inputs  the public key $h$ and $(m, c, \pi)$, and outputs 1/0 to accept/reject the statement that $m = \Dec_k(c)$.  
Beside,  we let $\VPKE$   to  satisfy the following extra properties (i.e., a specifically verifiable decryption):
\begin{itemize}
	\item {\em Completeness}. $\Pr[ \Verify_h(m, c, \pi)=1 \mid (m, \pi) \leftarrow \Prove_k(c)]= 1$, for  $\forall$   $c$  and $(h, k)\leftarrow\KGen(1^\lambda)$;
	\item {\em Soundness}.  For any  $(h, k)\leftarrow\KGen(1^\lambda)$ and    $c$, any  probabilistic polynomial-time (P.P.T.) $\huaA$ cannot produce  $\pi$   fooling  $\Verify_h$ to  accept that  $c$ is decrypted to $m'$  if $m' \ne \Dec_k(c)$,   with except negligible probability;
	\item {\em Zero-knowledge}. The proof $\pi$ can be simulated by a P.P.T. simulator $\huaS_\VPKE$, on input only public knowledge $m$, $h$ and $c$ that indeed satisfy $(m, c, h) \in \huaL_\VPKE := \{ \vec{x}:=(m, c, h) \mid m = \Dec_k(c) \wedge (h,k) \leftarrow \KGen(1^\lambda)\}$
\end{itemize}

\smallskip
\noindent
\textbf{Random oracle.} We treat the  cryptographic hash function as a global and programmable random oracle \cite{RO}  and denote the hash function with $\huaH$ through the paper.




\smallskip
\noindent
\textbf{Simulation-based paradigm.} 
To   formalize and prove   security,
a real world and an ideal world can be    defined and   compared: (i)  in the real world, there is an actual protocol $\Pi$  among the parties, some of which can be corrupted by an  adversary $\huaA$; (ii) in the ideal world, an ``imaginary'' trusted   ideal functionality $\huaF$  replaces the protocol and interacts with honest parties and a simulator $\huaS$. 
We say that
$\Pi$   {\em securely realizes} $\huaF$, if for
$\forall$   P.P.T. adversary $\huaA$ in the real-world, $\exists$ a P.P.T. simulator $\huaS$ in the ideal-world, s.t. the two
worlds cannot be distinguished, which means:
no P.P.T. distinguisher $\huaD$ can attain non-negligible advantage to distinguish ``the joint distribution  over the outputs of  
honest parties and  the adversary $\huaA$ in the real world'' from ``the joint distribution over the outputs of  
honest parties  and  the simulator $\huaS$ in the ideal world''.
%
%
%
%
%
%
%

Moreover, we consider the  static adversary who can corrupt some parties before the protocol starts.
%
The advantage of simulation-based paradigm is   that all  desired behaviors of the protocol can be precisely described by the  ideal functionality. Remarkably,   this approach has been widely adopted to analyze decentralized protocols \cite{KMS16,KZZ16,bentov2017instantaneous}   to capture the subtle adversary in the blockchain.



\section{Formalization of Decentralized \underline{H}uman \underline{I}ntelligent \underline{T}ask\underline{s}}\label{sec:problem}
This section   rigorously defines our security model, by giving the ideal functionality of   \underline{H}uman \underline{I}ntelligent \underline{T}ask\underline{s} (HITs)  that captures the   security/utility requirements of the state-of-the-art HITs    in  reality \cite{feifei,ng,video,Imagenet,ABI13, sdhc, mturkbotpanic,PWC15,mturkgolden,SZ15,imagenet-details,mturkbot,turkopticon,MCN16}. 
Our security modeling 
sets forth a clear security goal, that is: the HITs in the real world shall be as ``secure'' as the HITs in an admissible ideal world.


\smallskip
\noindent{\bf  Reviewing the HITs in reality.}
Let us briefly review the HITs adopted in reality
\cite{feifei,ng,video,Imagenet,ABI13, sdhc, mturkbotpanic,PWC15,mturkgolden,SZ15,imagenet-details,mturkbot,turkopticon,MCN16}, before presenting our abstraction of their ideal functionality.


\smallskip
\underline{\smash{\em {Parties \& process flow}}}.
There are two explicit roles in a HIT, i.e., the requester and some workers.\footnote{There is   an implicit registration authority (RA), who is   required by real-world crowdsourcing platforms e.g. MTurk to prevent adversary   forging a large number of identities (a.k.a. Sybil attackers). In practice, RAs can be instantiated by (i) the platform itself (e.g., MTurk), and (ii) the certificate authority who provides authentication service. Our solution can inherit these  established RAs, and we therefore omits such the  implicit RAs,  with assuming all identities are granted. If the participants   are interested in anonymity,    anonymous-yet-accountable authentication scheme   \cite{zebralancer,anonpass} can be used; however, those are orthogonal techniques out scope of this paper.} 
The \emph{requester}, uniquely identified by $\huaR$, can post a task $\huaT$ to collect a certain amount of answers. In the task, $\huaR$ also promises a concrete reward policy. 
The \emph{worker} with a unique identifier $\huaW_j$, submits his answer $\bm{a}_j$ to expect receive the   reward. 
%

\smallskip
\underline{\smash{\em {Task design}}}.
A HIT consists of a sequence of   questions denoted by $\huaT=( q_1,\cdots,q_N )$, where  each $q_i$ is a multiple choice question and $N$ is the number of questions in the task. The answer of each question must lay in a particular  $\mathsf{range} \subset \mathbb{N}\cup0$ pre-specified when $\huaT$ is published.

%

The above HIT design is based on batched choice questions, which follows real-world practices \cite{feifei,ng,video,Imagenet,ABI13, sdhc, mturkbotpanic,PWC15,mturkgolden,SZ15,imagenet-details,mturkbot,turkopticon,MCN16}   to remove  ambiguity, thus letting  workers precisely understand the task.
For example,  Fei-fei Li \emph{et al.} \cite{feifei,RL10,imagenet-details} used the technique  to create the   deep learning benchmark ImageNet, and Andrew Ng \emph{et al.} \cite{ng} suggested it for     language annotations.

\smallskip
\underline{\smash{\em {Answer quality}}}.
The quality of an answer   is induced by a function  $\reward(\bm{a}_j;sp)$, 
where $\bm{a}_j=( a_{(1,j)},\cdots,a_{(N,j)} )$ is the  answer submitted by worker $\huaW_j$, 
and $sp$ is some secret parameters of requester.
The output of $\reward(\cdot)$   is denoted by $\chi_j$, which is said to be the quality of worker  $\huaW_j$.

\ignore{
In general, the incentive mechanisms used in crowdsourcing can be seen as a function denoted by $\reward(\{\bm{a}_j\}_{i \in [1:K]};\aux)$, where $\{\bm{a}_j\}_{i \in [1:K]}$ is the set of all workers' submissions, and $\aux$ represents some auxiliary inputs (including budget) that parameterize the incentive mechanism \cite{zebralancer}.
In this paper, we consider the widespread state-of-the-art efforts on designing incentive mechanisms, and particularly,
focus on a broad array of elegant incentive design based on so-called gold-standards, as many impactful results and real-world experiences \cite{ng,SZ15,video,RL10,imagenet-details,mturkbot} have revealed its great effectiveness and graceful simplicity, through which we hopefully can establish a concrete decentralized crowdsourcing protocol that is practically efficient and useful.
}

The above abstraction    captures the quality-based incentive mechanism adopted by real-world HITs in Amazon's MTurk \cite{SZ15,imagenet-details,mturkgolden,mturkbot}.
For example,    a   task $\huaT$ consists of   $N$ questions, out of which   $M$ questions are golden-standard questions that are ``secretly'' mixed. 
The {\em quality} of a worker can be computed, due to her accuracy in the  $M$ golden-standard questions. 

Formally, in the qualify function $\reward(\bm{a}_j;sp)$, the parameter $sp=(G,G_S)$, where $G \subsetneq [1,N]$ represents the  randomly chosen  indexes of the  golden-standard questions, and $G_S=\{s_i|s_i\in\mathsf{range}\}_{i \in G}$ represents the known answers of the golden-standard questions. 
Following  the real-world practices \cite{SZ15,imagenet-details,mturkgolden,mturkbot}, the quality of an answer  $\bm{a}_j=( a_{(1,j)},\cdots,a_{(N,j)} )$  is:  
%
 $$\reward(\bm{a}_j, (G,G_S)) = {\sum_{i \in G}  [a_{(i,j)}  {\equiv}  s_i] }$$
where  $[\cdot]$ is Iverson bracket to  convert any  logic proposition to  1 if the proposition is true and 0 otherwise.


\smallskip
\noindent{\bf Defining the decentralized HITs' functionality.}
\label{sec2-3}
Now we are ready to present our security notion of HITs in the presence of cryptocurrency.
We  formalize the ideal functionality of HITs (denoted by $\huaF_{hit}$) in the $\huaL$-hybrid model  as shown in Fig \ref{fig:ideal-func}.
%
Intuitively,  $\mathcal{F}_{hit}^{\mathcal{L}}$   abstracts a special-purpose multi-party secure computation,
in which: (i) a requester  recruits $K$ workers to crowdsource some knowledge, and (ii) each worker gets a payment of $\bitcoin/K$ from the requester, if submitting an answer meeting the minimal quality standard $\Theta$.
 %

In greater detail, the ideal functionality $\huaF_{hit}$ of HITs immediately implies the following security properties: 

\ignore{
\smallskip
\noindent\underline{\emph{Special fairness}}.
In general, the ideal functionality $\mathcal{F}_{hit}^{\mathcal{L}}$ captures the special fairness in HITs requires that: the requester $\huaR$ must pre-specify the parameters (e.g. $pp$ and $sp$) of her incentive mechanism when her task is published, and later  $\huaR$ cannot pay the worker $\huaW_j$ less than the well-deserved reward  $r_j=\reward(\bm{a}_j;pp,sp)$, if and only if $\huaR$ indeed receives the answer $\bm{a}_j$; additionally, the public parameters $pp$ must be leaked to the public once the task is announced, and the secret parameters $sp$ cannot be leaked until the completion of the task.
Remark that $\mathcal{F}_{hit}^{\mathcal{L}}$  essentially captures the subtle adversarial network attacker that can maliciously schedule message deliveries in the blockchain, since the adversary shall not be able to break the fairness guarantees of honest parties, even if she corrupts the message ordering, the requester, or a portion of workers, or corrupt all these.
}



\begin{figure}
	\fbox{
		\parbox{.975\linewidth}{
			
			\begin{center}
				\large{\bf The ideal functionality of  HIT $\mathcal{F}_{hit}^{\mathcal{L}}$} 
			\end{center}

			\small
			Given accesses to oracle $\mathcal{L}$, the functionality $\mathcal{F}_{hit}^{\mathcal{L}}$ interacts with a requester $\huaR$, a set of workers  $\{\huaW_j\}$
			and adversary $\huaS$. 	
			
			\hrulefill{} {\bf Phase 1: Publish Task} \hrulefill{}
			
			\begin{itemize}[leftmargin=0.3cm]
				\item Upon receiving   $(\mathsf{publish}, N, \bitcoin, K, \mathsf{range},  {\Theta},  G, G_S )$ from $\huaR$, 
				leak $(\mathsf{publishing}, \huaR, N, \bitcoin, K, \mathsf{range}, \Theta, |G|, |G_S|)$ to $\huaS$, {\color{blue} until the beginning of next clock period, proceed with the following delayed executions}: 
				\begin{itemize}
					\item send $(\mathsf{freeze}, \mathcal{P}_i, \bitcoin)$ to $\huaL$, if return $(\mathsf{frozen}, \mathcal{F}_{hit}^{\mathcal{L}}, \mathcal{P}_i, \bitcoin)$:
					\begin{itemize}
						\item store $N$, $\bitcoin$, $K$, $\mathsf{Range}$, $\bar{\chi}$ and $sp$ as internal states;
						\item initialize   $\mathsf{answers}\leftarrow\emptyset$, and goto next phase; 
					\end{itemize}

				\end{itemize}
			\end{itemize}

			\hrulefill{} {\bf Phase 2:  Collect Answers} \hrulefill{}
			\begin{itemize}[leftmargin=0.3cm]
				\item Upon receiving $(\mathsf{answer}, \bm{a}_j)$ from $\huaW_j$, leak the message $(\mathsf{answering}, \huaW_j, |\bm{a}_j|)$ to $\huaS$, {\color{brown}  till receiving $(\mathsf{approved})$ from  $\huaS$, continue with the delayed executions down below}:
				\begin{itemize}
					\item if  $(\huaW_j,\cdot) \in \mathsf{answers}$, do nothing;
					\item else, $\mathsf{answers} \leftarrow \mathsf{answers} \cup (\huaW_j,\bm{a}_j)$, send $\mathsf{answers}$ to $\huaR$, leak $(\huaW_j,|\bm{a}_j|)$ to $\huaS$, go to phase 3 if $|\mathsf{answers}|=K$.
				\end{itemize}
			\end{itemize}

			\hrulefill{} {\bf Phase 3: Evaluate Answers} \hrulefill{}
			\begin{itemize}[leftmargin=0.3cm]
				\item Upon entering this phase, leak all received messages to $\huaS$, {\color{blue}  until the beginning of next clock period, proceed to run the following delayed executions} for   each $\mathcal{W}_j \in \{\mathcal{W}_j \mid   (\huaW_j,\cdot) \in \mathsf{answers}\}$:
				\begin{itemize}
					\item if receiving $(\mathsf{evaluate}, \mathcal{W}_j )$   from   $\huaR$, 
					proceed as: 
					\begin{itemize}
						\item  
						check whether $ \reward(\bm{a}_j, (G, G_S) ) \ge \Theta$, if that is the case,
						send $(\mathsf{pay}, \mathcal{W}_j, \bitcoin/K)$ to $\huaL$, 
						and   leak $(\mathsf{evaluated},\mathcal{W}_j, G, G_S)$ to all entities including  $\huaS$;

					\end{itemize}
					\item if receiving $(\mathsf{outrange}, \mathcal{W}_j, i)$   from $\huaR$, 
					proceed as: 
					\begin{itemize}
						\item  
						if   $a_{(i,j)} \notin \mathsf{range}$, leak $(\mathsf{outranged},\mathcal{W}_j, a_{(i,j)})$ to all entities, otherwise send $(\mathsf{pay}, \mathcal{W}_j, \bitcoin/K)$ to $\huaL$.
					\end{itemize}
					\item else, no  message from $\huaR$ was received, proceed as: 
					\begin{itemize}
						\item if   $\bm{a}_j \ne\bot$, send $(\mathsf{pay}, \mathcal{W}_j, \bitcoin/K)$ to $\huaL$.
					\end{itemize}

				\end{itemize}	
			\end{itemize}
			
			
			
			
			
			
			
		}
	}
	\caption{The (stateful) ideal functionality of coin-aided HIT $\mathcal{F}_{hit}^{\mathcal{L}}$. The {\color{blue} blue} text shows  $\mathcal{F}_{hit}^{\mathcal{L}}$ is proceeding synchronously as the adversary can delay message deliveries up to next clock period \cite{KMS16,KZZ16}; the {\color{brown} brown} text means that $\mathcal{F}_{hit}^{\mathcal{L}}$ has to proceed asynchronously as if the adversary can arbitrarily delay messages.}\label{fig:ideal-hit}
	\label{fig:ideal-func}
		 \vspace{-2mm}
\end{figure}

\begin{itemize}
\item
\underline{\smash{\emph{Fairness}}}. Our ideal functionality   captures a   strong notion of fairness, that means:
the worker get paid, if and only if s/he puts forth a qualified answer (instead of copying and pasting somewhere else).
In greater detail, the requester specifies a sequence of $N$ multi-choice questions, 
which are multi-choice questions having some options in $\mathsf{range}$ and contain $|G|$ gold-standard challenges.\footnote{We explicitly consider that $|G|$ and $\mathsf{range}$ are small constant in the HITs ideal functionality. Such   modeling follows real-world practices \cite{feifei,ng,video,Imagenet,ABI13, sdhc, mturkbotpanic,PWC15,mturkgolden,SZ15,imagenet-details,mturkbot,turkopticon,MCN16}.
In particular, $|\mathsf{range}|$ is a small constant in practice, because it represents few options of each multi-choice question in HIT;
and $|G|$ is also a small constant, as it represents few gold-standard challenges in a HIT task.}
For each worker, s/he has to 
(i) meet a pre-specified   quality standard $\Theta$  
and (ii) submit answers in the range of options,
in order to receive the pre-defined payment $\bitcoin/K$.


\item
\underline{\smash{\emph{Audibility of gold-standards}}}. 
The choice of golden standards is up to the requester, so it becomes a realistic worry that a malicious requester  uses some bogus as the answers of golden standard questions. 
%
%
%
The ideal functionality aims to abstract the best prior art \cite{MCN16,turkopticon} regarding this issue so far,
that means the golden standards become public auditable once the HIT is done. 
This abstraction ``simulates''  the ad-hoc reputation systems maintained by the MTurk workers  to grade the reputations of the MTurk requesters in reality \cite{MCN16,turkopticon}.

\item
\underline{\smash{\emph{Confidentiality}}}. 
It means     any worker cannot  learn the advantage information during the course of protocol execution. 
Without the property,    workers can   copy and paste     to free ride, 
so it is a minimal requirement to ensure  the usefulness of decentralized HITs. Our ideal functionality  naturally captures the property.

\end{itemize}

\noindent{\bf Adversary.} We consider   probabilistic polynomial-time adversary in the real world.
It can  corrupt the requester and/or some workers statically, before the real-world protocol begins. The  uncorrupted parties are said to be honest. 
Following the standard blockchain model \cite{KMS16,KZZ16}, we also abstract the ability of the real-world adversary to control the communication (between the blockchain and   honest parties) as:
(i) it follows the synchrony assumption  \cite{GKL15,KMS16}, namely, we let there is a global clock \cite{GKL15,KMS16}, and the adversary  can delay any messages sent to the blockchain up to a-priori known time (w.l.o.g., up to the next clock); (ii) the adversary can manipulate the order of so-far-undelivered messages sent to the blockchain, which is  known as  the ``rushing'' adversary. 
%

%
%
%


\smallskip
\noindent{\bf Expressivity of HITs' ideal functionality.}
The ideal functionality $\huaF_{hit}$   not only captures the elegant state-of-the-art of collecting image/language/video annotations \cite{feifei,RL10,imagenet-details,ng,mturkbot,video,SZ15}  but also reflects the common scenario of crowdsourcing human knowledge.
%
Consider the next example: Alice is running a small startup, and aims to provide a service to visualize the availabilities of street parkings. Unfortunately, at each moment, Alice   only knows the availabilities of street parkings  at quite few spots, since she cannot afford the cost of monitoring every corner around the city. The little a-priori knowledge of Alice is her ``golden standards'', and such information is too little to boost a useful service. 
So Alice can crowdsource more street parking information   from a few workers,  with using her few golden standards to control the quality of solicited data.

In light of the above discussion, it is fair to say that our abstraction  is expressive   to capture   most real-world practices of crowdsourcing human knowledge   (e.g.   HITs in MTurk).



\section{HITs Protocol and Security Analysis}
\label{protocol}
This section   elaborates  our   practical protocol for decentralized HITs. We begin with an important building block for proving the quality of encrypted answers.
Then we  showcase the smart contract functionality $\Contract_{hit}$ that interacts with the workers and the requester. Later, the detailed protocol is given in the presence of $\Contract_{hit}$. We finally prove that our protocol securely realizes the ideal functionality $\huaF_{hit}$ of HITs.



\subsection{Proof of   quality of encrypted answer ($\PoQoEA$)}
\label{sec:poqoea}

The core building block of our novel decentralized  protocol is to  allow  the requester {\em   efficiently prove the quality of encrypted answers}.
 We formally define this concrete purpose to set forth the notion of $\PoQoEA$, and then present  an efficient reduction from it to verifiable  decryption ($\VPKE$).

\smallskip
\noindent{\bf Defining $\PoQoEA$}.
The problem we are addressing here is to prove that: an encrypted answer $\bm{c}_j$ can be decrypted to obtain some $\bm{a}_j$ s.t. the quality of $\bm{a}_j$ is $\chi$, without leaking anything other than $\bm{c}_j$, $\chi$ and the parameters of quality function.

To capture the problem,
the state-of-the-art \cite{zebralancer,KMS16}   adopts the standard notion of zk-proof in order to support generic quality measurements. 
Different from   existing solutions, we particularly tailor the notion of zk-proof to obtain a fine-tuned notion of $\PoQoEA$ for the widely adopted quality function
defined in \S \ref{sec:problem}. Namely, we consider $\reward(\cdot~; G, G_s)$ where $G$ is the index of gold-standards and $G_s=\{s_i\}_{i \in G}$ is the ground truth of golden standards, and aim to remove the unnecessary generality in the concrete setting. 

Precisely, given the quality function $\reward(\cdot~; G, G_s)$ and any established public key encryption scheme $(\Enc_h, \Dec_k)\leftarrow\KGen(1^\lambda)$, we can define $\PoQoEA$ as   a tuple of hereunder algorithms $(\cProve_k, \cVerify_h)$: 
\begin{enumerate}
	\item $\cProve_k({\bf c}_j, \chi, G, G_s ) \rightarrow  \pi$. Given the encrypted answer ${\bf c}_j = ( c_{1,j},\dots,c_{N,j} )$, the quality $\chi$, and the golden standards $(G,$ $G_s)$, it outputs a proof  $\pi$ attesting $\chi$ is the quality of ${\bf c}_j$; the algorithm explicitly takes the secret decryption key $k$ as input;
	\item $\cVerify_h({\bf c}_j, \chi, \pi, G, G_s) \rightarrow 0/1$. 
	It outputs 0 (reject) or 1 (accept), according to whether $\pi$ is a valid proof attesting $\chi$ is the actual quality of ${\bf c}_j$; the algorithm explicitly takes the public encryption key $h$ as input;
\end{enumerate}

Moreover, $\PoQoEA$  shall satisfy the following properties:
\begin{itemize}
	\item \underline{\smash{\emph{Completeness}}}. $\PoQoEA$ is   complete, if for any $G$, $G_s$, ${\bf c}_j$, $\chi$ and $(\Enc_h, \Dec_k)$ s.t. 
	$\chi = \reward(\Dec_k({\bf c}_j); G, G_s)$, there is 
	$\Pr[\cVerify_h({\bf c}_j, \chi, \pi, G, G_s) = 1 \mid   \pi \leftarrow \cProve_k({\bf c}_j, \chi, G, G_s)]=1$;
	\item \underline{\smash{\emph{``{Upper-bound}'' soundness}}}. $\PoQoEA$ is  upper-bound sound, if for any $G$, $G_s$, ${\bf c}_j$, $\chi$ and $(\Enc_h, \Dec_k)$,  for $\forall$ P.P.T.     $\huaA$,  there is
	$\Pr[ \cVerify_h({\bf c}_j, \chi, \pi', G, G_s) = 1 \wedge \chi<\reward(\bm{a}_j;G, G_s) \wedge {\bf a}_j=\Dec_k({\bf c}_j) \mid \pi'  \leftarrow \huaA(G, G_s, \chi, {\bf c}_j,\lambda,\Enc_h,\Dec_k)]\le negl(\lambda)$, where  $negl(\lambda)$ is a negligible function in   $\lambda$; so  it is computationally infeasible   to produce a valid proof, if  $\chi$ is not the    upper bound of the  quality of what ${\bf c}_j$ is encrypting;
	\item \underline{\smash{\emph{{``Special'' zero-knowledge}}}}. 
	Conditioned on $|G|$ and the range of elements in $G_s$ are   small constants,
	for any $G$, $G_s$, ${\bf c}_j$, $\chi$ and $(\Dec_k, \Enc_h)$, 
	$\exists$ a  P.P.T. simulator $\huaS$ that can simulate  the communication scripts of $\PoQoEA$ protocol on input only   $h$, $G$, $G_s$, ${\bf c}_j$, and $\chi$.
	

\end{itemize}

\smallskip
\noindent{\bf Rationale behind the finely-tuned abstraction}.
The   notion of $\PoQoEA$ is defined to remove needless generality  in the special  case of HITs. 
Compared to the state-of-the-art  notion \cite{zebralancer}, $\PoQoEA$  
is more promising to be  efficiently constructed,
as it brings the following definitional advantages:
\begin{itemize}
	\item We adopt \emph{``upper-bound'' soundness} to ensure that any probably corrupted requester cannot forge the upper bound of quality of each worker. 
	Such the tuning stems from a basic fact that: the reward of a worker is   an increasing function in quality, 
	so the upper bound of the worker's quality at least reflects the well-deserved reward of the worker. As a result, any cheating requester has to pay at least as much as the honest requester.
	\item Another major difference is   the relaxed \emph{{special zero-knowledge}}: $\PoQoEA$ is zero-knowledge, only if $|G|$  and $\mathsf{range}$ are   small constants, so  anything simulatable by the gold standards can be leaked.
	Nevertheless, the  conditions are prevalent in the special context of HITs \cite{feifei,ng,video,Imagenet,ABI13, sdhc, mturkbotpanic,PWC15,mturkgolden,SZ15,imagenet-details,mturkbot,turkopticon,MCN16}. Recall that  $G$  represents the  few golden standard questions,  and   $\mathsf{range}$   means the few options of each  question in HITs, indicating that both are small constants in reality.
\end{itemize} 

In sum, even though  $\PoQoEA$ is seemingly over-tuned,
it essentially coincides with the generic zk-proof of the quality of encrypted answers  in the context of   HITs.

\smallskip
\noindent{\bf Construction and security analysis.} Here is an efficiency-driven way to constructing  $\PoQoEA$  for the quality function $\reward(\bm{a}_j; G, G_s)$ that was defined in \S \ref{sec:problem}.
We can reduce the problem to the standard notion of verifiable decryption.
More precisely, 
given the established  $\VPKE$ scheme $(\Enc_h, \Dec_k, \Prove_k, \Verify_h)$,
$\PoQoEA$ can be constructed as illustrated in Fig \ref{fig:poqoea}.

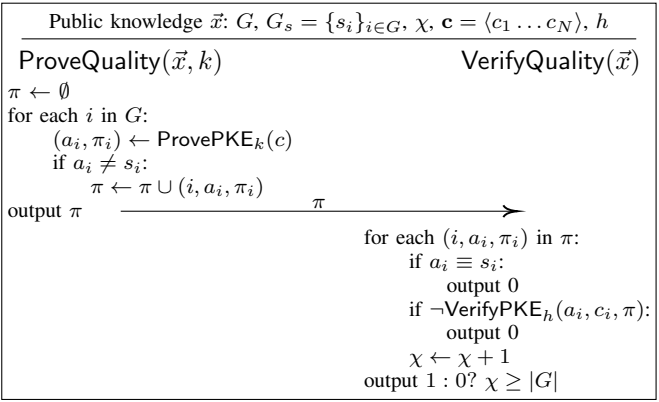
\begin{figure}[!htb] 
	
	 \vspace{-3mm}
	
	\begin{tikzpicture}[x=0.75pt,y=0.75pt,yscale=-1,xscale=1]
	
	\draw   (100,95) -- (430,95) -- (430,295) -- (100,295) -- cycle ;

	\draw    (160,200) -- (360,200) ;
	\draw [shift={(360,200)}, rotate = 180] [color={rgb, 255:red, 0; green, 0; blue, 0 }  ][line width=0.75]    (7,-3) .. controls (6,-1) and (3,-0.3) .. (0,0) .. controls (3,0.3) and (6,1) .. (7,3)   ;

	\draw (160,125) node    {$\cProve(\vec{x},k)$};
	\draw (377,125) node    {$\cVerify(\vec{x})$};

	\draw (265,105) node   [font=\footnotesize] {Public knowledge $\vec{x}$: $G$, $G_{s}=\{s_i\}_{i\in G}$, $\chi$, ${\bf c}=\langle c_{1} \dots c_{N} \rangle$, $h$};
	\draw    (110,113) -- (420,113) ;
	
	\draw (175,170) node   [font=\footnotesize][align=left] {
		$\pi$ $\leftarrow$ $\emptyset$\\
		for each $i$ in $G$:\\
		\hspace{5mm} $(a_i, \pi_i) \leftarrow \Prove_k(c) $\\
		\hspace{5mm} if $a_i \ne s_i$: \\
		\hspace{10mm} $\pi \leftarrow \pi \cup (i,a_i,\pi_i)$ \\
		output $\pi$};
	\draw (355,250) node   [font=\footnotesize][align=left] {for each $(i,a_i,\pi_i)$ in $\pi$:\\
		\hspace{5mm} if $a_{i} \equiv s_i$: \\
		\hspace{10mm} output 0 \\
		\hspace{5mm} if $\neg\Verify_h(a_i,c_{i},\pi)$: \\
		\hspace{10mm} output 0 \\
		\hspace{5mm} $\chi \leftarrow \chi + 1$   \\
		output $1: 0 ?$ $\chi\ge |G|$ };
	\draw (260,196) node   [font=\footnotesize] {$\pi$};

	\end{tikzpicture}
	
	\caption{The construction of $\PoQoEA$  for the quality     defined in \S \ref{sec:problem}.}
	\label{fig:poqoea}
	\vspace{-1.5mm}
\end{figure}

\begin{lemma}
	Given any verifiable public key encryption $\VPKE$,
	the algorithm in Fig \ref{fig:poqoea} satisfies the definition of $\PoQoEA$ regarding the quality function defined in \S \ref{sec:problem}.
\end{lemma}
\begin{proof}
{\em (sketch)} 
The completeness is immediate from the definition of quality function, the correctness of encryption, and the completeness of  $\VPKE$.
To prove  the upper-bound soundness, we assume by contradiction to let  an adversary break it, then the adversary can immediately  break the soundness of $\VPKE$. The special zero-knowledge is also clear: considering $|G|$ and the $\mathsf{range}$ of each $a_i$ are constants,
the permutation $\binom{|G|}{\chi}$  would be constant, indicating that there exists a P.P.T. simulator $\huaS$  invoking at most polynomial number of $\huaS_\VPKE$ (on input $c_i$, $h$, and    guessed  $a_i \in \mathsf{range} \setminus \{s_i\}$) to simulate all $\VPKE$ proofs \cite{lindellsimulate}, thus simulating the $\PoQoEA$ proof.
\end{proof}

\begin{figure}[!htb] 
	\fcolorbox{black}{Azure}{
		\parbox{.95\linewidth}{
			\begin{center}
				\large{\bf The HITs contract functionality \color{black}{$\Contract_{hit}^{\mathcal{L}}$}} 
			\end{center}
			\small
			Given accesses  to $\mathcal{L}$,  $\Contract_{hit}$ interacts with   $\huaR$,   $\{\huaW_j\}$,
			and  $\huaA$. 	
			
			\hrulefill{} {\bf Phase 1: Publish Task} \hrulefill{}
			
			\begin{itemize}[leftmargin=0.3cm]
				\item Upon receiving   $(\mathsf{publish}, N, \bitcoin, K, \mathsf{range},  {\Theta}, h, \comm_{gs})$ from $\huaR$, leak the message and $\huaR$ to $\huaA$,
				{\color{blue} until the beginning of next clock, proceed with the delayed executions down below}: 
				\begin{itemize}[leftmargin=0.3cm]
					\item  
					send $(\mathsf{freeze}, \mathcal{P}_i, \bitcoin)$ to $\huaL$,
					if returns $(\mathsf{frozen}, \mathcal{F}_{hit}^{\mathcal{L}}, \mathcal{P}_i, \bitcoin)$:
					\begin{itemize}[leftmargin=0.3cm]
						\item store $N$, $\bitcoin$, $K$, $\mathsf{range}$,  ${\Theta}$, $h$   and $\comm_{gs}$
						\item initialize $\mathsf{answers}$  $\leftarrow$ $\emptyset$, $\mathsf{comms}$ $\leftarrow$ $\emptyset$
						\item send $(\mathsf{published}, \huaR, N, \bitcoin, K, \mathsf{range},  {\Theta}, h, \comm_{gs})$ to all entities, and goto phase 2-a
					\end{itemize}
				\end{itemize}
			\end{itemize}

			\hrulefill{} {\bf Phase 2-a:  Collect Answers (Commit phase)} \hrulefill{}
			\begin{itemize}[leftmargin=0.3cm]
				\item Upon receiving    $(\mathsf{commit}, \comm_{\bm{c}_j})$ from $\huaW_j$, 
				leak the message and $\huaW_j$ to $\huaA$,
				{\color{blue} then proceed with the following delayed executions until the beginning of next clock}, with consulting $\huaA$ to re-order all received $\mathsf{commit}$ messages: 
				\begin{itemize}[leftmargin=0.3cm]
					\item for each received  $\mathsf{commit}$ message (sent from $\huaW_j$):
					\begin{itemize}[leftmargin=0.3cm]
						\item if $(\huaW_j, \cdot) \notin \mathsf{comms}$ and $(\cdot, \comm_{\bm{c}_j}) \notin \mathsf{comms}$:
						\begin{itemize}[leftmargin=0.3cm]
							\item let $\mathsf{comms} \leftarrow \mathsf{comms} \cup (\huaW_j, \comm_{\bm{c}_j})$
							\item if $|\mathsf{comms}| = K$, send $(\mathsf{committed},\mathsf{comms})$ to all entities, and goto the $\mathsf{reveal}$  phase
						\end{itemize}
					\end{itemize}
				\end{itemize}
				
			\end{itemize}

			\hrulefill{} {\bf Phase 2-b:  Collect Answers (Reveal phase)} \hrulefill{}
			\begin{itemize}[leftmargin=0.3cm]
				\item Upon entering this phase, leak all received messages and their senders to $\huaA$, {\color{blue}  till the next clock period, proceed as}: 
				\begin{itemize}[leftmargin=0.3cm]
					\item for   each $\mathcal{W}_j \in \{\mathcal{W}_j \mid   (\huaW_j,\cdot) \in \mathsf{comms}\}$:
					\begin{itemize}[leftmargin=0.3cm]
						\item if receiving the message $(\mathsf{reveal}, \bm{c}_j, \mathsf{key}_j)$ from $\huaW_j$ such that  $\open(\comm_{\bm{c}_j}, \bm{c}_j, \mathsf{key}_j)=1$:
						\begin{itemize}[leftmargin=0.3cm]
							\item  $\mathsf{answers}$ $\leftarrow$ $\mathsf{answers} \cup (\huaW_j,\bm{c}_j)$
						\end{itemize}
						\item else $\mathsf{answers}$ $\leftarrow$ $\mathsf{answers} \cup (\huaW_j,\bot)$
					\end{itemize}
					\item 
					send $(\mathsf{revealed},\mathsf{answers})$ to all,  and  goto the next  phase 
				\end{itemize}
			\end{itemize}
			
			\hrulefill{} {\bf Phase 3: Evaluate Answers} \hrulefill{}
			\begin{itemize}[leftmargin=0.3cm]
				\item Upon entering this phase, leak all received messages and their senders to $\huaA$, {\color{blue}  till the next clock period, proceed as}: 
				
				\begin{itemize}[leftmargin=0.3cm]
					\item if receiving $(\mathsf{golden}, G, G_s, \mathsf{key}_{gs})$  from   $\huaR$, such that $\open(\mathsf{comms}_{gs},G||G_s,\mathsf{key}_{gs})=1$:
					\begin{itemize}[leftmargin=0.3cm]
						\item for   each $\mathcal{W}_j \in \{\mathcal{W}_j \mid   (\huaW_j,\cdot) \in \mathsf{answers}\}$:
						\begin{itemize}[leftmargin=0.3cm]
							\item if receiving $(\mathsf{outrange}, \mathcal{W}_j, i, a_{(i,j)}, \pi_i)$   from $\huaR$:
							\begin{itemize}[leftmargin=0.3cm]
								\item[]  send $(\mathsf{pay}, \mathcal{W}_j, \bitcoin/K)$ to $\huaL$, if    $a_{(i,j)} \in \mathsf{range}$ or $\Verify_h(a_{(i,j)}, c_{(i,j)},  \pi_i) = 0$
							\end{itemize}
							\item else if receiving   $(\mathsf{evaluate}, \mathcal{W}_j, \chi_j, \pi)$  from $\huaR$:
							\begin{itemize}[leftmargin=0.3cm]
								\item[]  send $(\mathsf{pay}, \mathcal{W}_j, \bitcoin/K)$ to $\huaL$, if  $\chi_j \ge \Theta$ or $\cVerify_h({\bf c}_j, \chi_j, \pi, G, G_s) = 0$
							\end{itemize}
							\item else if $\bm{c}_j \ne\bot$, send $(\mathsf{pay}, \mathcal{W}_j, \bitcoin/K)$ to $\huaL$

						\end{itemize}
					\end{itemize}
					\item otherwise, for each $\mathcal{W}_j \in \{\mathcal{W}_j \mid   (\huaW_j,\cdot) \in \mathsf{answers}\}$, send $(\mathsf{pay}, \mathcal{W}_j, \bitcoin/K)$ to $\huaL$
				\end{itemize}
				
			\end{itemize}
			
		}
	}
	\caption{The ideal functionality of the (stateful) HITs contract.}
	\label{fig:contract}
	 \vspace{-4mm}
\end{figure}

\subsection{HIT contract   and HIT protocol}
\label{sec:contract}

Now we are ready to present our concretely efficient decentralized  protocol $\Pi_{hit}$ for HIT.
%
Our design centers around a smart contract $\Contract_{hit}^\huaL$,
which is formally described in Fig \ref{fig:contract}. The contract  $\Contract_{hit}^\huaL$ is the crux to take best advantage of the rather limited abilities of blockchain to make our protocol securely realize the ideal functionality $\huaF_{hit}^\huaL$.
Thus given contract $\Contract_{hit}^\huaL$, our HITs protocol  $\Pi_{hit}$ can be defined  among the requester, the worker and the contract, as formally illustrated in Fig \ref{fig:protocol}.
%
Informally, our HIT protocol $\Pi_{hit}$   proceeds as follows:
\begin{enumerate}
	\item  \underline{\smash{\emph{Publish task}}}. The requester $\huaR$ announces her public key $h$, and publishes a   task  $\huaT$  of $N$ multi-choice questions to crowdsource $K$ answers for the task.  Each question in  $\huaT$ is  specified to have some options in  $\mathsf{range}$. 
	The task    mixes some golden standard questions, whose indexes $G$ and ground truth $G_s$ are committed to $\comm_{gs}$.
	Also, $\huaR$  places $\bitcoin$ as deposit to cover her budget, which promises that a worker would get a reward of $\bitcoin/K$, if submitting an answer   beyond a specified quality standard $\Theta$.
	\item  \underline{\smash{\emph{Commit answers}}}. Once the task is published, the workers can commit their answers (encrypted to the requester) in the task. To prevent against copy-and-paste attacks, duplicated commitments are rejected. The contract moves   to the next phase, once $K$ distinct workers commit.
	
	\item  \underline{\smash{\emph{Reveal answers}}}. After $K$   workers commit their answers, these workers can start to   reveal their answers
	in form of ciphertexts   encrypted to the requester.
	Note  that the submissions of answers explicitly contain  two subphases, namely, committing and revealing, 
	which is the crux to prevent the network adversary from taking advantages by adversarially scheduling the order of submissions.
	
	\item  \underline{\smash{\emph{Evaluate answers}}}. Eventually,  the requester is supposed to instruct the blockchain  to correctly pay these encrypted answers to facilitate the critical fairness. 
	To this end, the protocol leverages our novel notion of $\PoQoEA$.	So the requester can efficiently prove to the contract to reject a certain answer, if the worker  does not meet the pre-specified quality standard $\Theta$.
	If an answer is out of the  specified   $\mathsf{range}$, the requester   is   allowed to use verifiable encryption $\VPKE$ to reveal that to reject payment.
\end{enumerate}

\underline{\smash{\emph{Remark}}}.  $\Contract_{hit}^{\mathcal{L}}$ captures the essence of smart contracts \cite{Woo14} in reality, as it:
(i)  reflects the transparency of Turing-complete smart contract that is a stateful program handling pre-specified tasks publicly;
(ii)  captures   a     contract that can access   the cryptocurrency ledger to honestly    deal with conditional payments;
(iii)  models  the    network adversary who is consulted to schedule the delivering order of  so-far-undelivered messages.

\begin{figure}
	\fcolorbox{black}{Ivory1}{
		\parbox{.95\linewidth}{
			\begin{center}
				\large{\bf The protocol of   HITs \color{black}{$\Pi_{hit}$}} 
			\end{center}	
		\small			
			 $\Pi_{hit}$ is among the requester $\huaR$, the workers $\{\huaW_j\}$ and     $\Contract_{hit}$
			
			\hrulefill{} {\bf Phase 1: Publish Task} \hrulefill{}
			\begin{itemize}[leftmargin=0.3cm]
				\item Requester $\huaR$: 
				\begin{itemize}[leftmargin=0.3cm]
					\item  $(\Enc_h, \Dec_k) \leftarrow \KGen(1^\lambda)$				
					\item Upon receiving the parameters $G$, $G_s$, $\Theta$, $N$, $\mathsf{range}$,  $\bitcoin$, $K$ of a HIT to publish:
					\begin{itemize}[leftmargin=0.3cm]
						\item $\mathsf{key}_{sg}$ $\overset{\$}{\leftarrow}$ $\{0,1\}^\lambda$
						\item $\comm_{gs} \leftarrow \commit( G||G_s, \mathsf{key}_{sg})$ 
						\item send $(\mathsf{publish}, N, \bitcoin, K, \mathsf{range},  {\Theta}, h, \comm_{gs})$ to $\Contract_{hit}$
					\end{itemize}					 
				\end{itemize}
			\end{itemize}
			
			\hrulefill{} {\bf Phase 2:  Collect Answers} \hrulefill{}
			\begin{itemize}[leftmargin=0.3cm]
				\item Worker $\huaW_j$: 
				\begin{itemize}[leftmargin=0.3cm]
					\item Upon receiving $(\mathsf{published}, \huaR, N, \bitcoin, K, \mathsf{range},  {\Theta}, h, \comm_{gs})$ from $\Contract_{hit}$:
					\begin{itemize}[leftmargin=0.3cm]
						\item  get   the answer  $\bm{a}_j  = (a_{(1,j)},\cdots,a_{(N,j)})$
						\item $\bm{c}_j \leftarrow (\Enc_h(a_{(1,j)}),\cdots,\Enc_h(a_{(1,N)}))$
						\item $\comm_{\bm{c}_j} \leftarrow \commit(\bm{c}_j,\mathsf{key}_{j})$, where $\mathsf{key}_{j}$ $\overset{\$}{\leftarrow}$ $\{0,1\}^\lambda$
						\item send $(\mathsf{commit}, \comm_{\bm{c}_j})$ to $\Contract_{hit}$
					\end{itemize} 
					\item Upon receiving $(\mathsf{committed},\mathsf{comms})$ from $\Contract_{hit}$:
					\begin{itemize}[leftmargin=0.3cm]
						\item if $(\huaW_j,\cdot) \in \mathsf{comms}$, send $(\mathsf{reveal}, \bm{c}_j, \mathsf{key}_j)$ to $\Contract_{hit}$
					\end{itemize}
				\end{itemize}
			\end{itemize}
			
			\hrulefill{} {\bf Phase 3: Evaluate Answers} \hrulefill{}
			\begin{itemize}[leftmargin=0.3cm]
				\item Requester $\huaR$:
				\begin{itemize}[leftmargin=0.3cm]
					\item Upon receiving $(\mathsf{revealed},\mathsf{answers})$ from $\Contract_{hit}$:
					\begin{itemize}[leftmargin=0.3cm]
						\item send $(\mathsf{golden}, G, G_s, \mathsf{key}_{gs})$  to   $\huaR$
						\item for each $(\huaW_j, \bm{c}_j)$ $\in$ $\mathsf{answers}$:
						\begin{itemize}[leftmargin=0.3cm]
							\item decrypt each item in $\bm{c}_j$ to get  $\bm{a}_j=(a_{(1,j)},\cdots,a_{(N,j)})$
							\item if $\exists a_{(i,j)} \in \bm{a}_j$ s.t. $a_{(i,j)} \notin \mathsf{range}$:
							\begin{itemize}
								\item $ (a_{(i,j)}, \pi_i) \leftarrow \Prove_k(c_{(i,j)})$
								\item send $(\mathsf{outrange}, \mathcal{W}_j, i, a_{(i,j)}, \pi_i)$ to $\Contract_{hit}$
							\end{itemize}
							\item else if $\chi_j = \reward(\Dec(\bm{c}_j,sk_\huaR);G,G_s) < \Theta$:
							\begin{itemize}
								\item $\pi \leftarrow \cProve_k({\bf c}_j, \chi_j, G, G_s)$
								\item send $(\mathsf{evaluate}, \mathcal{W}_j, \chi_j, \pi)$ to   $\Contract_{hit}$
							\end{itemize}
	
						\end{itemize}

					\end{itemize}

				\end{itemize}
			\end{itemize}
		}
	}
	\caption{The formal description of the decentralized HITs protocol $\Pi_{hit}$.}
	\label{fig:protocol}
	\vspace{-3.3mm}
\end{figure}


\subsection{Instantiating cryptographic building blocks}
For sake of completeness, we  hereafter give  the constructions of   cryptographic building blocks.
Let  $\huaG= \langle g \rangle$ be  a cyclic group of prime order $p$, where $g$ is a random  generator of $\huaG$.

\smallskip
\underline{\smash{\emph{(Short $\mathsf{range}$) verifiable decryption}}} is   based on exponential ElGamal. The private key $k \overset{\$}{\leftarrow} \mathbb{Z}_p$, the public key $h = g^k$, the encryption     $\Enc_h(m)=(c_1,c_2) =(g^r,g^m{h}^r)$,  and the decryption   $\Dec_k((c_1, c_2))= \log(c_2 / c_1^k)$  where $\log$ is to brute-force the short plaintext $\mathsf{range}$ to obtain $m$; if decryption fails to output  $m\in\mathsf{range}$, then $c_2 / c_1^k$ is returned.
%
In addition, to efficiently augment the above $(\Enc_h, \Dec_k)$  to be {\em verifiable}, we adopt a variant of Schnorr protocol \cite{Sch89} (for Diffie-Hellman tuples) with Fiat-Shamir transform  in the random oracle model. In detail,
\begin{itemize} 
	\item $\Prove_k((c_1,c_2))$. Run $\Dec_k((c_1, c_2))$ to obtain $m \in \mathsf{range}$ (or $g^m$ if $m \notin \mathsf{range}$). Let $x \overset{\$}{\leftarrow} \{0,1\}^{\lambda}$. Compute $A=c_1^x$, $B=g^x$, $C=\huaH(A||B||g||h||c_1||c_2||g^m)$,  $Z=x+kC$, and $\pi=(A,B,Z)$.  If $m \in \mathsf{range}$, output $(m, \pi)$; else, output $(g^m, \pi)$.
	\item $\Verify_h(M,(c_1,c_2),\pi)$. Parse $\pi=(A,B,Z)$. If $M \in \mathsf{range}$, compute   $C'=\huaH(A||B||g||h||c_1||c_2||g^M)$,  and then verify $g^{M \cdot C'} \cdot c_1^Z  {\equiv} A \cdot c_2^{C'}$ and $g^Z \equiv B\cdot h^{C'}$, output 1 if the verification passes and 0 otherwise; else if $M \in \huaG$, compute  $C'=\huaH(A||B||g||h||c_1||c_2||M)$  and  verify $M^{C'} \cdot c_1^Z  {\equiv} A \cdot c_2^{C'} \wedge g^Z \equiv B\cdot h^{C'}$, output 1 iff the verification passes and 0 otherwise.
\end{itemize}

\smallskip
\underline{\smash{\emph{Proof of   quality of encrypted answer}}} is built by invoking the  above $\VPKE$ construction in a black-box manner, due to our reduction from $\PoQoEA$  to $\VPKE$ in \S \ref{sec:poqoea}.

\smallskip
\underline{\smash{\emph{Commitment scheme}}}
is   instantiated  according to the well-known efficient folklore construction in the random oracle model \cite{RO,fairswap}:
	(i) $\commit(\mathsf{msg},\mathsf{key}) = \huaH(\mathsf{msg}||\mathsf{key})$;  
	(ii) $\open(\comm,\mathsf{msg}',\mathsf{key}') = [\huaH(\mathsf{msg}'||\mathsf{key}') \equiv \comm]$, where  $[\cdot]$ is Iverson bracket from
	a proposition to  1 (true) or 0 (false).  



\subsection{Security analysis}
\label{sec:analysis}

\begin{theorem}
	Conditioned on the hardness of DDH problem and static corruptions, the stand-alone instance of $\Pi_{hit}$ securely realizes $\huaF_{hit}$ in  $\Contract_{hit}^\huaL$-hybrid, random oracle model.
\end{theorem}
\noindent{\em Proof.} (sketch)
Let $\mathbb{C}$ denote the set of corrupted parties controlled by  the adversary $\huaA$,  and let  $\mathbb{H}$ denote the set of rest honest parties.
For any P.P.T. adversary $\huaA$ in the real world, we can sketch a P.P.T. simulator  $\huaS$  in the ideal world  to interact with the ideal functionality $\huaF_{hit}$ and   corrupted parties, such that   $\huaS$ can  emulate  the actions of  honest parties and the contract $\Contract_{hit}$. 
Detailedly, $\huaS$ proceeds as follows:
\begin{itemize}
	 
	\item \textbf{Publish Task (Phase 1)}.  
	If $\huaR\in\mathbb{C}$, considering that the corrupted $\huaR$ sends the $\mathsf{publish}$ message to $\Contract_{hit}$ in the real world, $\huaS$ can trivially simulate that with interacting with $\huaF_{hit}$. If  $\huaR\in\mathbb{H}$, when  the honest $\huaR$ sends the $\mathsf{publish}$ message to $\huaF_{hit}$, $\huaS$ is informed and thus  allows $\huaS$ to simulate the  real-world scripts of {\em publish task}.
	  
	\item \textbf{Collect Answers (Phase 2)}. 
	In the real world, the P.P.T. adversary $\huaA$ might: (i) corrupt a set of parties $\mathbb{C}$ up to including the requester and a set of the workers, and (ii) is also consulted to reorder the so-far-undelivered messages sent to $\Contract_{hit}$ (till the next clock). 
	
	The basic strategy to emulate $\huaA$ is that: $\huaS$ invokes the adversary $\huaA$ to obtain how $\huaA$ is re-ordering the $\mathsf{commit}$ messages (sent from  workers), let $\mathbb{W}$ to represent the set of workers whose $\mathsf{commit}$ messages are scheduled as the first $K$ to deliver; then $\huaS$ delays all $\mathsf{answer}$ messages that are not sent from the workers in  $\mathbb{W}$. Then, $\huaS$ internally simulates the ciphertexts sent via $\mathsf{reveal}$ messages to open commitments. If $\huaR\in\mathbb{H}$, the ciphertexts can be simulated as they are indistinguishable from the uniform distribution over the ciphertext space; if $\huaR\in\mathbb{C}$,  $\huaS$ is informed about all answer submissions sent from the workers, thus can internally simulate the submissions of the workers in the real world.
	
	Moreover, if   $\huaA$ corrupts a worker whose $\mathsf{commit}$ message is scheduled in the first $K$ to deliver but does not send any the $\mathsf{reveal}$ message to open the commitment, the simulator $\huaS$ can simulate that since it can let the corrupted worker to send an $\mathsf{answer}$ message containing $\bot$ with $\huaF_{hit}$.
	In addition, it is trivial to see $\huaS$ can internally simulate the parties as well as $\Contract_{hit}$, when the adversary $\huaA$ corrupts a worker to submit duplicated commitment.
	
	%

	\item \textbf{Evaluate Answers (Phase 3)}. The simulation becomes clear, if considering the security requirements of commitment scheme, $\VPKE$, and $\PoQoEA$. 
	If  the requester $\huaR\in\mathbb{C}$, the simulator $\huaS$ invokes $\huaA$ to obtain all $\mathsf{outrange}$ and/or $\mathsf{evaluate}$ messages sent to $\Contract_{hit}$, and then simulates the interactions.
	If  the requester $\huaR\in\mathbb{H}$, whenever $\huaR$ sends $\mathsf{outrange}$ and/or $\mathsf{evaluate}$ messages to $\huaF_{hit}$, $\huaS$ is informed and hence is allowed to simulate the interactions between $\Contract_{hit}$ and $\huaR$ in the real world.
\end{itemize}



\section{\system: Implementation \& Evaluation}\label{implementation}

To demonstrate the feasibility of our   protocol, we implement it to build {\system}, and then use the system to launch a typical image annotation task for ImangeNet \cite{imagenet-details,RL10} atop Ethereum.

 \smallskip
\noindent{\bf System overview.}
\system{} consists of an on-chain part and an off-chain part: the on-chain  smart contract is deployed in Ethereum ropsten network; the requester client and worker clients are implemented in Python 3.6. The off-chain clients are installed in a PC that uses Ubuntu 14.04 LTS and equips Intel Xeon E3-1220V2 CPU and 16 GB main memory.

 \smallskip
 \underline{\smash{\emph{ImageNet's HIT task}}}.
We demonstrate our system through an ImageNet task \cite{imagenet-details,RL10}, which is specified as: each task is made of 106 binary questions, 100 out of which are non-gold-standard questions, while the remaining 6 questions are requester's gold-standard challenges; 4 workers are allowed to participate; if a worker cannot correctly answer at least four golden standard questions, his submission will be rejected without being paid, otherwise he deserves to get the payment.

 \smallskip
 \underline{\smash{\emph{Cryptographic modules}}}.
The hash function  is instantiated by keccak256. We choose  the cyclic group $\huaG$ by   using the $\huaG_1$ subgroup of BN-128 elliptic curve, over which all   concrete public key primitives are instantiated. 

 \smallskip
 \underline{\smash{\emph{Code availability}}}.
The code of our prototype is available at \url{https://github.com/njit-bc/dragoon}.
An experiment instance is  atop Ethereum ropsten network   (\url{https://ropsten.etherscan.io/address/0x5481b096c78c8e09c1bfbf694e934637f7d66698}).

\begin{figure}[!htp]
	\vspace{-4mm}
	\centering
	\includegraphics[width=7cm]{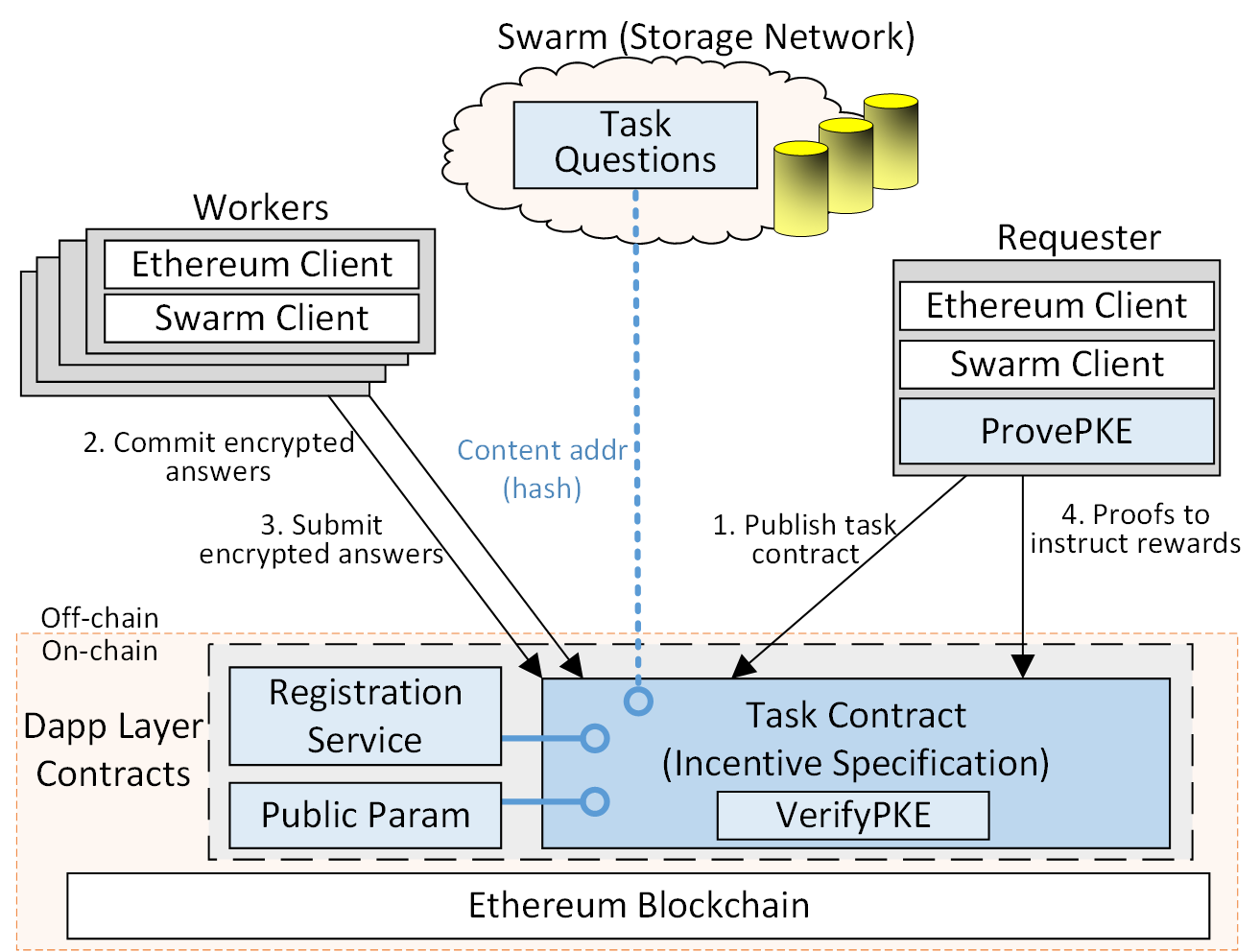}
	\caption{The schematic diagram of \system{} at a high-level.  
	}
	\label{fig:system}
	\vspace{-2mm}
\end{figure}

 \smallskip
\noindent{\bf Implementation details.}
Many non-trivial on- and off-chain optimizations are particularly made  for practicability.

 \smallskip
 \underline{\smash{\emph{Off-chain ends}}}.
The requester end   warps: (i) an Ethereum node to interact with the blockchain, e.g. publish task, download workers' submissions, etc; (ii) the prover of verifiable encryption to generate necessary proofs to instruct the contract   to reward workers;
(iii) a Swarm API to publish the detailed questions of each crowdsourcing task.   Swarm \cite{Swarm} is an off-chain storage network, where the questions of HIT is stored; in addition, to ensure integrity of HIT questions, the  digest of the questions is committed in the contract, which significantly reduces on-chain cost, without violating   securities. 

The worker client wraps Ethereum to     interact with the blockchain to read task and submit answers, and also incorporates Swarm client to allow download task questions.

 \smallskip
\underline{\smash{\emph{On-chain optimizations}}}.
We  cautiously perform a few non-trivial system-level optimizations to lighten the  task contract: (i) we implement all public key schemes over   $\huaG_1$ subgroup of BN-128 \cite{bncurve}, since we can use some precompiled contracts in Ethereum to do algebraic operations there cheaply \cite{EIP1108}; (ii)    it is expensive to   store   ciphertexts in   the contract as internal variables, so we make the contract   store their 256-bit hashes instead and let the actual ciphertexts included in the chain as emitted event logs \cite{Woo14}.

 \smallskip
\noindent{\bf Evaluations.}
We conduct intensive experiments to measure the concrete performance, and discuss the system feasibilities from the on-chain side and the off-chain side.

 \smallskip
 \underline{\smash{\emph{Off-chain costs}}}.
First, \system{} enables the requester to manage only one private-public key pair throughout all her tasks,  because all protocol scripts are simulatable without secret key and therefore leak nothing relevant. 
%
%
More importantly, the off-chain   cost of proving relevant cryptographic proofs   is significantly reduced by removing unnecessary generality.

\begin{table}[!htbp]
	\vspace{-3mm}
	\scriptsize
 
	\centering
	\caption{ Off-chain {\bf Proving Cost} of $\VPKE$ and $\PoQoEA$ due to our concrete constructions and generic zk-proofs respectively.}
		\vspace{-2mm}
	\label{table:off-chain}
				\renewcommand{\arraystretch}{1.2}
	\begin{tabular}{c||cm{0.8cm}m{1.4cm}m{1.1cm}}
		\hline
		& Statement to Prove & Time & Peak Memory		
		\\ \hline\hline
		\multirow{2}{*}{Ours} & $\VPKE$  & 3 ms      &  53 MB           
		\\
								  & $\PoQoEA$  & 10 ms      &  53 MB            
								  \\ \hline
		\multirow{2}{*}{Generic ZKP}{$*$}   & $\VPKE$  & 37 s         & 3.9 GB             
		\\
								  &	$\PoQoEA$  & 112 s        & 10.3 GB             
								  \\ \hline
	\end{tabular}

{{\ }\\ * Through our evaluations, generic zk-proofs are instantiated by zk-SNARK, which is the only generic zk-proof feasibly supported by existing blockchains to our knowledge.}
		 \vspace{-1mm}
\end{table}

Table \ref{table:off-chain} clarifies   the requester suffers from hindersome off-chain burden of generating generic zk-proofs.
The concrete construction  removes the bottleneck of proving in generic zk-proof. 
First, the requester can generate a proof to reject a worker's submission within only a few milliseconds, 
which costs nearly 2 minutes if using generic zk-proof. Second, the concretely efficient constructions also save in memory usage. For example, by generic zk-proof, rejecting a submission requires a peak memory usage of 10 GB, which is reduced to only 53 MB by concrete constructions.

 \smallskip
 \underline{\smash{\emph{On-chain costs}}}.
We  measure the critical on-chain performance from many angles including the cost of verifying zk-proofs and the on-chain gas usage of the whole protocol.

First, we compare the verifying cost of concrete and generic constructions for $\VPKE$ and $\PoQoEA$ (six golden standards) in Table \ref{table:on-chain-verif}. The     concrete proofs  are faster, even compared to the generic zk-proof (SNARK)   known for   efficient verification. 

%
%
%

 \newcommand{\minitab}[2][l]{\begin{tabular}{#1}#2\end{tabular}} 
\begin{table}[!htbp]
	\vspace{-3mm}
	\scriptsize
	
	\centering
	\caption{On-chain {\bf Verification Cost} of $\VPKE$ and $\PoQoEA$ due to our concrete constructions and generic zk-proofs, respectively.}
	\vspace{-2mm}
	\label{table:on-chain-verif}
	\renewcommand{\arraystretch}{1.2}
	\begin{tabular}{c ||cm{1.6cm}}
		\hline
		& Statement to Verify                                      & Verifying Time    \\ \hline\hline
		\multirow{2}*{{Ours}}$ \dagger$  & $\VPKE$     & {\ \ \ \ }1 ms       \\
		& $\PoQoEA$                			& {\ \ \ \ }2 ms       \\ \hline
		\multirow{2}*{Generic ZKP}$ \ddagger$   &$\VPKE$ & {\ \ \ \ }11 ms  	      \\
		& $\PoQoEA$                 		& {\ \ \ \ }17 ms          \\ \hline
	\end{tabular}

	{{\ }
	\\ {$\dagger$} Here the   implementation of BN-128 curve is  from   \url{https://github.com/scipr-lab/libff}.
	\\ {$\ddagger$} The evaluations for generic ZKP (SNARK) are performed due to constructions from 2048-bit RSA-OAEP over $\mathbb{Z}_{pq}$ instead of ElGamal over the $\huaG_1$ subgroup of BN-128.}
	\vspace{-1mm} 
\end{table}

Moreover, the overall handling fee  of running a concrete ImageNet instance is summarized in Table \ref{table:on-chain}. 
To estimate the cost of on-chain usages, we apply a gas price at $1.5\times10^{-9}$ Ether per gas, and an Ether price at $115$ USD per Ether, which are the safe-low price of gas  \cite{ethgasstation} and the market  price of Ether on March/17th/2020,  respectively.
Under the above exchange rate, the on-chain handling fee paid by  each worker is about \$0.48, which is used to submit an answer. 
In addition, thanks to the efficient verification of $\PoQoEA$, the requester can   spend {\em few cents}  to reject each low-quality answer.
%
The overall on-chain handling cost of the entire HIT is about two US dollars. 
In contrast, when  MTurk facilitates the same ImageNet task, it charges a handling fee at least \$4 currently \cite{mturkprice,mturkfee}.

\begin{table}[!htbp]
	\vspace{-3mm}
	\scriptsize
 
	\centering
	\caption{On-chain {\bf Overall Handling Fees} of a Concrete ImageNet Task   \protect\\ (Task policy: 4 workers; 106 questions; 6 gold-standards; a submission is rejected if failing in 3 gold-standards)}
	\label{table:on-chain}
		\vspace{-2mm}
			\renewcommand{\arraystretch}{1.2}
	\begin{tabular}{m{4.25cm}||m{1.2cm}||m{0.79cm}}
		\hline
		Handling fee of & Gas Usage &  In USD\\
		\hline\hline
		Publish task (by requester) 				 &  $\sim$1293 k      	&  \$0.22    	\\
		Submit answers (by worker) 					 &  $\sim$2830 k    	&  \$0.48        	\\
		Verify $\PoQoEA$ to reject an answer    	 &  $\sim$180 k   		&  \$0.03		\\    		
		\hline\hline
		Overall (best-case: reject no submission)    &  $\sim$12164 k       &  \$2.09      	\\
		Overall (worst-case: reject all submissions) &  $\sim$12877 k       &  \$2.22       	\\
		\hline
	\end{tabular}

 \vspace{-1mm} 
\end{table}

To summarize, \system{} is practical. 
Our experiment even reveals that \system{}'s on-chain handling cost can  be economically cheaper than the the handling fee charged by third-party platforms such as MTurk. 
In addition,    \system{}  is compatible with many alternative  chains (e.g.,  Cardano, Ethereum Classic  and more) other than Ethereum, 
as long as the   blockchains are using Ethereum Virtual Machine (EVM) as the running environment of smart contracts.
So our system can be deployed in these alternative chains to further reduce the   handling cost. 





\section{Conclusion}\label{sec6}
We design a  decentralized protocol to   crowdsource human knowledge, which, to   our knowledge, is  the first system in its kind that realizes both rigorous security and high efficiency.

  \smallskip
\noindent{\bf Open problems.} 
It hints that the special-purpose protocols are promising to
decentralize  various  crowdsourcing  with high-security assurance as well as efficiency.
It immediately corresponds to a few realistic problems  to explore.
For example, can we design a concretely efficient protocol
to decentralize participatory crowd-sensing that is  minimally meaningful with the needed fairness and privacy?
Such the problem is   challenging, 
since there seems no explicit requester to ``prove'' the quality of encrypted data anymore.
Unfortunately, letting the blockchain   learn encrypted data's quality (without a prover)   falls into the category of (multi-input) functional encryption,
which is   unclear how can be solved practically till today. 

Another fundamental problem is that we consider  security  due to   conventional cryptographic notions, 
where   corrupted parties are fully controlled by an adversary and   honest parties    follow  the protocol independently.
The model has an inherent drawback to explain why  rational workers would not deviate (e.g., by colluding). To resolve the concern, an ``incentive-compatible'' protocol is required, so ``following the protocol'' is a   Nash equilibrium (or its refinement) that can deter rational workers from deviating.





\section*{Acknowledgment}
Qiang is supported in part by JD.com and a Google Faculty Award. We would like to thank the anonymous reviewers for their valuable suggestions and  comments about this paper.

\bibliographystyle{IEEEtran}
{\footnotesize
\bibliography{IEEEabrv,reference}
}

\end{document}